\documentclass[11pt]{article}

\usepackage{times}
\usepackage{latexsym}
\usepackage{amsfonts,amsthm,amssymb}
\usepackage{euscript}
\usepackage{amstext}
\usepackage{graphicx}
\usepackage{color}
\usepackage{url,hyperref}
\usepackage{fullpage}

\newtheorem{lemma}{Lemma}[section]
\newtheorem{theorem}[lemma]{Theorem}

\newtheorem{definition}[lemma]{Definition}
\newtheorem{corollary}[lemma]{Corollary}

\newtheorem{remark}[lemma]{Remark}

        {\hspace*{\fill}$\Box$\par}

\newcommand{\opt}{\textrm{\sc OPT}}
\newcommand{\etal}{et al.\ }
\newcommand{\eps}{\epsilon}
\newcommand{\Algorithm}[1]{{\texttt{\bf{#1}}}} 
\newcommand{\lwf}{\Algorithm{LWF}}

\newcommand{\one}{\texttt{Type1}}
\newcommand{\two}{\texttt{Type2}}
\newcommand{\grdy}{\Algorithm{LF}}
\newcommand{\len}{\mathfrak{\rho}}

\newcommand{\floor}[1]{\lfloor #1 \rfloor}
\newcommand{\ceil}[1]{\lceil #1 \rceil}

\newcommand{\fe}{F^e}
\newcommand{\fl}{F^l}

\newcommand{\lwfs}{\lwf^{S}}

\newcommand{\lwfne}{\lwf^{N^e}}
\newcommand{\lwfneone}{\lwf^{N_1^e}}
\newcommand{\lwfnetwo}{\lwf^{N_2^e}}
\newcommand{\lwfnl}{\lwf^{N^l}}
\newcommand{\opts}{\opt^{S}}
\newcommand{\optn}{\opt^{N}}

\newcommand{\cJ}{{\cal J}}
\newcommand{\spone}{-3.5mm}

\newcommand{\spbox}{-4mm}

\begin{document}

\title{Longest Wait First for Broadcast Scheduling}
\author{
Chandra Chekuri\thanks{Department of Computer Science, University of Illinois, 201 N.\ Goodwin Ave., Urbana, IL 61801.
{\tt chekuri@cs.uiuc.edu}. Partially supported by NSF grants CCF-0728782
 and CNS-0721899. }
 \and Sungjin Im\thanks{Department of Computer Science, University of
Illinois, 201 N.\ Goodwin Ave., Urbana, IL 61801. {\tt im3@uiuc.edu}} \and Benjamin Moseley\thanks{Department of
Computer Science, University of Illinois, 201 N.\ Goodwin Ave., Urbana, IL 61801. {\tt bmosele2@uiuc.edu}. Partially
supported by NSF grant CNS-0721899. } }
\date{\today}
\maketitle

\begin{abstract}
  We consider {\em online} algorithms for broadcast scheduling. In the
  pull-based broadcast model there are $n$ unit-sized pages of
  information at a server and requests arrive online for pages. When
  the server transmits a page $p$, all outstanding requests for that
  page are satisfied. There is a lower bound of $\Omega(n)$ on the
  competitiveness of online algorithms to minimize average flow-time;
  therefore we consider resource augmentation analysis in which the
  online algorithm is given extra speed over the adversary. The {\em
    longest-wait-first} (LWF) algorithm is a natural algorithm that
  has been shown to have good empirical performance
  \cite{AksoyF98}. Edmonds and Pruhs showed that LWF is $6$-speed
  $O(1)$-competitive using a very complex analysis; they also showed
  that LWF is not $O(1)$-competitive with less than $1.618$-speed. In  this paper we make two main contributions to the analysis of LWF and
  broadcast scheduling.
  \begin{itemize}
  \item We give an intuitive and easy to understand analysis of LWF
    which shows that it is $O(1/\eps^2)$-competitive for average
    flow-time with $(4+\eps)$ speed. Using a more involved analysis, we
    show that LWF is $O(1/\eps^3)$-competitive for average flow-time
    with $(3.4+\epsilon)$ speed.
  \item We show that a natural extension of LWF is $O(1)$-speed
    $O(1)$-competitive for more general objective functions such as
    average delay-factor and $L_k$ norms of delay-factor (for fixed
    $k$). These metrics generalize average flow-time and $L_k$ norms
    of flow-time respectively and ours are the first non-trivial
    results for these objective functions in broadcast scheduling.
  \end{itemize}
\end{abstract}


\newpage
\setcounter{page}{1}

\section{Introduction}
We consider online algorithms for broadcast scheduling in the
pull-based model. In this model there are $n$ pages (representing some
form of useful information) available at a server and clients request
a page that they are interested in. The server {\em broadcasts} pages
according to some online policy and {\em all} outstanding requests for
a page are satisfied when that page is transmitted/broadcast. This is
what distinguishes this model from the standard scheduling models
where the server has to process each request separately. Broadcast
scheduling is motivated by several applications. Example situations
where the broadcast assumption is natural include wireless and
satellite networks, LAN based systems and even some multicast
systems. See \cite{Wong88,AcharyaFZ95,AksoyF98,Hall03} for pointers to
applications and systems that are based on this model. In addition to
their practical interest, broadcast scheduling has been of much
interest in recent years from a theoretical point of view. There is by
now a good amount of literature in online and offline algorithms in
this model \cite{BarnoyBNS98,AksoyF98,AcharyaFZ95,BartalM00,Hall03}. There is
also substantial work in the stochastic and queuing theory literature
\cite{DebS73,Deb84,Weiss79,WeissP81} on related models which make
distributional assumptions on the request arrivals. In a certain
sense, $\lwf$ can be shown to be optimal when page arrivals are
independent and assumed to have a Poisson distribution
\cite{AmmarW85}.


It is fair to say that algorithmic development and analysis for
broadcast scheduling have been challenging even in the simplest
setting of unit-sized pages; so much so that a substantial amount of
technical work has been devoted to the development of {\em offline}
approximation algorithms
\cite{KalyanasundaramPV00,ErlebachH02,GandhiKKW04,GandhiKPS06,BansalCKN05,BansalCS06};
many of these offline algorithms are non-trivial and are based on
linear programming based methods. Further, most of these offline
algorithms, with the exception of \cite{BansalCS06}, are in the
resource augmentation model of Kalyanasundaram and Pruhs
\cite{KalyanasundaramP95} in which the analysis is done by giving the
algorithm a machine with speed $s > 1$ when compared to a speed $1$
machine for the adversary.  In this paper we are interested in {\em
  online} algorithms in the worst-case competitive analysis
framework. We consider the problem of minimizing average flow-time (or
waiting time) of requests and other more stringent objective
functions. It is easy to show an $\Omega(n)$ lower bound on the
competitive ratio \cite{KalyanasundaramPV00} of any deterministic
algorithm and hence we also resort to resource augmentation analysis.
For average flow-time three algorithms are known to be
$O(1)$-competitive with $O(1)$-speed.  The first is the natural {\em
  longest-wait-first} ($\lwf$) algorithm/policy: at any time $t$ that
the server is free, schedule the page $p$ for which the total waiting
time of all outstanding requests for $p$ is the largest. Edmonds and
Pruhs \cite{EdmondsP04}, in a complex and original analysis, showed
that $\lwf$ is a $6$-speed $O(1)$-competitive algorithm and also that
it is not $O(1)$-competitive with a speed less than
$(1+\sqrt{5})/2$; they also conjectured that their lowerbound is tight. The same authors also gave a different algorithm
called BEQUI in \cite{EdmondsP03} and show that it is a
$(4+\eps)$-speed $O(1)$-competitive algorithm; although the algorithm
has intuitive appeal, the proof of its
performance relies on an involved reduction to an algorithm for a
non-clairvoyant scheduling problem \cite{Edmonds00} whose analysis
itself is substantially complex. The recent improved result in
\cite{EdmondsP09} for the non-clairvoyant problem when combined with
the reduction mentioned above leads to a $(2+\eps)$-speed
$O(1)$-competitive algorithm; however the new algorithm requires the
knowledge of $\epsilon$ and hence is not as natural as the other
algorithms. The preemptive algorithms in \cite{EdmondsP03,EdmondsP09} are also
applicable when the page sizes are arbitrary; see
\cite{Hall03} for empirical evaluation in this model.
At a technical level, a main difficulty in online analysis for broadcast
scheduling is the fact shown in \cite{KalyanasundaramPV00} that no
online algorithm can be {\em locally}-competitive with an adversary\footnote{An algorithm is locally-competitive if at each time $t$, its queue size
is comparable to that of the queue size of the adversary. Many results in
standard scheduling are based on showing local-competitiveness.}.

We focus on the $\lwf$ algorithm in the setting of unit-sized
pages. In addition to being a natural greedy policy, it has been shown
to outperform other natural policies \cite{AksoyF98}; moreover,
related variants are known to be optimal in certain stochastic
settings.  It is, therefore, of interest to better understand its
performance. We are motivated by the following questions. Is there a
simpler and more intuitive analysis of $\lwf$ for broadcast scheduling
than the analysis presented in \cite{EdmondsP04}?  Can we close the
gap between the upper and lower bounds on the speed requirement of
$\lwf$ to guarantee constant competitiveness?  Can we obtain
competitive algorithms for more stringent objective functions than
average flow-time such as $L_k$ norms of flow-time, average
delay-factor\footnote{Delay-factor is a recently introduced metric and
  we describe it more formally later.}  and $L_k$ norms of
delay-factor? We give positive answers to these questions.

\medskip
\noindent {\bf Results:} Our results are for unit-size pages. We make two
contributions.
\begin{itemize}
\item We give a simple and intuitive analysis of $\lwf$ that already
  improves the speed bound in \cite{EdmondsP04}; the analysis shows
  that $\lwf$ is $(4+\eps)$-speed $O(1/\epsilon^2)$-competitive for
  average flow time. Using a more complex analysis, we show that
  $\lwf$ is $(3.4+\eps)$-speed $O(1/\epsilon^3)$-competitive.
\item We show that a natural generalization of $\lwf$ that we call
  $\grdy$ is $O(k)$-speed $O(k)$-competitive for minimizing the $L_k$
  norm of flow time --- these bounds extend to average delay
  factor and $L_k$ norms of delay factor. These are the first
  non-trivial results for $L_k$ norms in broadcast scheduling for $k >
  1$.
\end{itemize}
$L_k$ norms for flow-time for some small $k >1$ such as $k=2,3$ have
been suggested as alternate and robust metrics of performance; see
\cite{BansalP03,Pruhs07} for more on this. Our results show that
$\lwf$-like algorithms have reasonable theoretical performance even
for these more difficult metrics. We derive these additional results
in a unified fashion via a general framework that is made possible by
our simpler analysis for $\lwf$. In our recent work \cite{ChekuriIM09}
we show that $\grdy$ is not $O(1)$-competitive with any constant speed
for the $L_\infty$-norm of delay factor. This suggests that $\grdy$
may require a speed that increases with $k$ to obtain
$O(1)$-competitiveness for $L_k$ norms. We note that the algorithms in
\cite{EdmondsP03,EdmondsP09} that perform well for average flow time
do not easily extend to the more general objective functions that we
consider.

Our analysis of $\lwf$ borrows several key ideas from
\cite{EdmondsP04}, however, we make some crucial simplifications. We
outline the main differences in Section~\ref{sec:overview} where we
give a brief overview of our approach.

\medskip
\noindent {\bf Notation and Formal Definitions:} We assume that
the server has $n$ distinct unit-sized pages of
information.  We use $J_{p,i}$ to denote $i$'th request for a page
$p \in \{1, \ldots, n\}$.  We let $a_{p,i}$ denote the arrival
time of the request $J_{p,i}$. The finish time $f_{p,i}$ of a
request $J_{p,i}$ under a given schedule/algorithm is defined to
be the earliest time after $a_{p,i}$ when the page $p$ is
sequentially transmitted by the scheduler; to avoid notational
overload we assume that the algorithm is clear from the context.
Note that multiple requests for the same page can have the same
finish time. The total flow time for an algorithm over a sequence
of requests is now defined as $\sum_p \sum_i (f_{p,i} - a_{p,i})$.
Delay-factor is a recently introduced metric in scheduling
\cite{ChangEGK08,BenderCT08,ChekuriM09}.  In the context of
broadcast scheduling, each request $J_{p,i}$ has a soft deadline
$d_{p,i}$ that is known upon its arrival. The slack of $J_{p,i}$
is $d_{p,i} - a_{p,i}$. The delay-factor of $J_{p,i}$ with finish
time $f_{p,i}$ is defined to be $\max(1, \frac{f_{p,i} -
a_{p,i}}{d_{p,i} - a_{p,i}})$; in other words it is the ratio of
the waiting time of the request to its slack. It can be seen that
delay-factor generalizes flow-time since we can set $d_{p,i} =
a_{p,i} + 1$ for each (unit-sized) request $J_{p,i}$. Given a
scheduling metric such as flow-time or delay-factor that, for each
schedule assigns a value $m_{p,i}$ to a request $J_{p,i}$, one can
define the $L_k$ norm of this metric in the usual way as
$\sqrt[k]{\sum_{(p,i)} m_{p,i}^k}$. Note that minimizing the sum
of flow-times or delay-factors is simply the $L_1$ norm problem.
In resource augmentation analysis, the online algorithm is given a
faster machine than the optimal offline algorithm. For $s \ge 1$,
an algorithm $A$ is $s$-speed $r$-competitive if $A$ when the
given $s$-speed machine achieves a competitive ratio of $r$.

In this paper we assume, for simplicity, the discrete time model. In
this model, at each integer time $t$, the following things happen
exactly in the following order; the scheduler make a decision of which
page $p$ to broadcast; the page $p$ is broadcast and all outstanding
requests of page $p$ are immediately satisfied, thus having finish
time $t$; new requests arrive. Note that new pages which arrive at $t$
are not satisfied by the broadcasting at the time $t$. It is important
to keep it in mind that all these things happen only at integer times.
See \cite{EdmondsP04} for more discussion on discrete time versus
continuous time models. For the most part, we assume for simplicity
of exposition, that the algorithm is given an integer speed $s$ which
implies that the algorithm schedules (at most) $s$ requests in each
time slot.  For this reason we present our analysis for $5$-speed and
$4$-speed which extend to $(4+\eps)$-speed and $(3.4+\eps)$-speed
respectively.  Due to space constraints we defer the details of the
extensions.

\noindent {\bf Our Recent Results:} Very recently, we
have proved that $\lwf$ is $(2.777 + \eps)$-speed $O(1)$-competitive. 
Also we have shown that $\lwf$ is
not $O(1)$-competitive with $(2- \eps)$-speed for any $\eps>0$, 
which disproves Edmonds and Pruhs's conjecture that 
their lowerbound $(1 + \sqrt{5})/2$
is tight \cite{EdmondsP04}.
The analysis of these new results is complex and
built on the work of this paper. 
Due to insufficient space, we will
include these new results in the full version. 

\medskip
\noindent {\bf Related Work:} We give a very brief description of
related work due to space constraints. We refer the reader to the
survey on online scheduling by Pruhs, Sgall and Torng \cite{PruhsST}
for a comprehensive overview of results and algorithms (see also
\cite{Pruhs07}).  Broadcast scheduling has seen a substantial amount
of research in recent years; apart from the work that we have already
cited we refer the reader to
\cite{KalyanasundaramPV00,CharikarK06,KhullerK04}, the recent paper of
Chang \etal \cite{ChangEGK08}, and the surveys \cite{PruhsST,Pruhs07}
for several pointers to known results.  As we mentioned already, a
large amount of the work on broadcast scheduling has been on offline
algorithms including NP-hardness results and approximation algorithms
(often with resource augmentation). With few exceptions
\cite{EdmondsP03}, almost all the work has focused on the unit-size
page assumption. Apart from the work on average flow-time that has
been mentioned before, the other work on online algorithms for
flow-time are the following. Bartal and Muthukrishnan
\cite{BartalM00,ChangEGK08} showed that the first-come-first-serve
rule (FCFS) is $2$-competitive for maximum flow-time. More recently,
Chekuri and Moseley \cite{ChekuriM09} developed a $(2+\eps)$-speed
$O(1/\eps^2)$-competitive algorithm for maximum delay-factor; we note
that this algorithm requires knowledge of $\epsilon$. Constant
competitive online algorithms for maximizing throughput are given in
\cite{Kimc04,ChanLTW04,ZhengFCCPW06,ChrobakDJKK06}. Algorithms to
minimize $L_k$ norms of flow-time in the context of standard
scheduling have been studied in \cite{BansalP03} and
\cite{ChekuriGKK04}. \vspace{\spone}

\subsection{Overview of Analysis}\vspace{-2mm}
\label{sec:overview} We give a high level overview of our analysis of
$\lwf$.  Let $\opt$ denote some fixed optimal 1-speed offline
solution; we overload notation and use $\opt$ also to denote the value
of the optimal schedule. Recall that for simplicity of analysis, we
assume the discrete-time model in which requests arrive at integer
times. For the same reason we analyze $\lwf$ with an integer speed $s
> 1$. We can assume that $\lwf$ is never idle. Thus, in each time step
$\lwf$ broadcasts $s$ pages and the optimal solution broadcasts $1$
page. We also assume that requests arrive at integer times.  At time
$t$, a request is in the set $U(t)$ if it is unsatisfied by the
scheduler at time $t$. In the broadcast setting $\lwf$ with speed $s$
is defined as the following. \vspace{\spbox}\vspace{-2mm}

\begin{center}
\begin{tabular}[r]{|c|}
\hline
\textbf{Algorithm}: $\lwf_s$ \\

\begin{minipage}{\textwidth}
\begin{itemize}
\item At any integer time $t$, broadcast the $s$ pages with the
  largest waiting times, where the waiting time of page $p$ is
  $\sum_{J_{p,i} \in U(t)} (t-a_{p,i})$.
\end{itemize}
\vspace{-7mm}
\end{minipage}\\\\

\hline
\end{tabular}
\end{center}\vspace{-2mm}

Our analysis of $\lwf$ is inspired by that in \cite{EdmondsP04}.  Here we summarize our approach and indicate the main
differences from the analysis in \cite{EdmondsP04}. Given the schedule of $\lwf_s$ on a request sequence $\sigma$, the
requests are partitioned into two disjoint sets $S$ (self-chargeable requests) and $N$ (non-self-chargeable requests).
Let the total flow time accumulated by $\lwf_s$ for requests in $S$ and $N$ be denoted by $\lwf_s^{S}$ and $\lwf_s^{N}$
respectively. Likewise, let $\opt^S$ and $\opt^N$ be the flow-time $\opt$ accumulates for requests in $S$ and $N$,
respectively. $S$ is the set of requests whose flow-time is comparable to their flow-time in $\opt$. Hence one
immediately obtains that $\lwf_s^{S} \le \len \opt^S$ for some constant $\len$.  For requests in $N$, instead of
charging them only to the optimal solution, these requests are charged to the total flow time accumulated by $\lwf$
{\em
  and} $\opt$. It will be shown that $\lwf_s^{N} \leq \delta \lwf_s +
\len \opt^N$ for some $\delta < 1$; this is crux of the proof. It follows that $\lwf_s = \lwf_s^{S} + \lwf_s^{N} \leq
\len \opt^S + \len \opt^N + \delta \lwf \leq \len \opt + \delta \lwf$. This shows that $\lwf_s \leq
\frac{\len}{1-\delta} \opt$, which will complete our analysis. Perhaps the key idea in \cite{EdmondsP04} is the idea of
charging $\lwf_s^N$ to $\lwf_s$ with a $\delta < 1$; as shown in \cite{KalyanasundaramPV00}, no algorithm for any
constant speed can be locally competitive with respect to all adversaries and hence previous approaches in the
non-broadcast scheduling context that establish local competitiveness with respect to $\opt$ cannot work.

In \cite{EdmondsP04}, the authors do not charge $\lwf_s^N$ directly to $\lwf_s$. Instead, they further split $N$ into
two types and do a much more involved analysis to bound the flow-time of the type 2 requests via the flow-time of type
1 requests. Moreover, they first transform the given instance to canonical instance in a complex way and prove the
correctness of the transformation. Our simple proof shows that these complex arguments can be done away with. We also
improve the speed bounds and generalize the proof to other objective functions. \vspace{\spone}

\subsection{Preliminaries}
To show that $\lwf_s^{N} \leq \delta \lwf_s + \rho \opt^N$, we will map the requests in $N$ to other requests scheduled
by $\lwf_s$ which have comparable flow time.  An issue that can occur when using a charging scheme is that one has to
be careful not to overcharge.  In this setting, this means for a single request $J_{p,i }$ we must bound of the number
of requests in $N$ which are charged to $J_{p,i}$.  To overcome the overcharging issue, we will appeal to a
generalization of Hall's theorem.  Here we will have a bipartite graph $G = (X \cup Y)$ where the vertices in $X$ will
correspond to requests in $N$.  The vertices in $Y$ will correspond to all requests scheduled by $\lwf_s$.  A vertex $u
\in X$ will be adjacent to a vertex $v \in Y$ if $u$ and $v$ have comparable flow time and $v$ was satisfied while $u$
was in our queue and unsatisfied; that is, $u$ can be charged to $v$.  We then use a simple generalization of Hall's
theorem, which we call \emph{Fractional Hall's Theorem}. Here a vertex of $u \in X$ is matched to a vertex of $v \in Y$
with weight $\ell_{u,v}$ where $\ell_{u,v}$ is not necessarily an integer. Note that a vertex can be matched to
multiple vertices.
\vspace{-2mm}

\begin{definition} [$c$-covering]
  \label{def:covering} Let $G = (X \cup Y, E)$ be a bipartite graph
  whose two parts are $X$ and $Y$, and let $\ell : E \rightarrow
  [0,1]$ be an edge-weight function.  We say that $\ell$ is a
  $c$-covering if for each $u \in X$, $\sum_{(u,v) \in E} \ell_{u,v} =
  1$ and for each $v \in Y$, $\sum_{(u,v) \in E} \ell_{u,v} \leq c$.
\end{definition}

The following lemma follows easily from either Hall's Theorem or the Max-Flow Min-Cut Theorem.
\vspace{-1mm}
\begin{lemma}[Fractional Hall's theorem]
 \label{lem:Hall_thm}
 Let $G = (V = X \cup Y, E)$ be a bipartite graph whose two parts are
 $X$ and $Y$, respectively. For a subset $S$ of $X$, let $N_G(S) = \{v
 \in Y| uv \in E $, $ u \in S\}$, be the neighborhood of $S$. For
 every $S \subseteq X$, if $ |N_G(S)| \geq {1 \over c}|S|$, then there
 exists a $c$-covering for $X$.
\end{lemma}

Throughout this paper we will discuss time intervals and unless explicitly mentioned we will assume that they are
closed intervals with integer end points. When considering some contiguous time interval $I=[s,t]$ we will say that
$|I|=t-s+1$ is the length of interval $I$; in other words it is the number of integers in $I$. For simplicity, we abuse
this notation; when $X$ is a set of closed intervals, we let $|X|$ denote the number of distinct integers in some
interval of $X$. Note that $|X|$ also can be seen as the sum of the lengths of maximal contiguous sub-intervals if $X$
is composed of non-overlapping intervals.

To be able to apply Lemma~\ref{lem:Hall_thm}, we show another lemma which will be used throughout
this paper. Lemma~\ref{lem:intervals} says that the union of some fraction of time intervals is comparable to that of the whole
time interval.

\begin{lemma}
\label{lem:intervals} Let $X = \{[s_1,t_1], \ldots, [s_k,t_k]\}$ be a finite set of closed intervals and let $X' =
\{[s'_1,t_1], \ldots, [s'_k,t_k]\}$ be an associated set of intervals such that for $1 \le i \le k$, $s'_i \in
[s_i,t_i]$ and $|[s'_i,t_i]| \ge \lambda |[s_i,t_i]|$. Then $|X'| \ge \lambda |X|$.
\end{lemma}

\vspace{-6mm}

\section{Minimizing Average Flow Time}\vspace{-2mm}
\label{sec:broadcast} \vspace{-2mm} We focus our attention to minimizing average flow time. A fair amount of notation
is needed to clearly illustrate our ideas.  Following \cite{EdmondsP04}, for each page, we will partition time into
intervals via \emph{events}.  Events for page $p$ are defined by $\lwf_s$'s broadcasts of page $p$. When $\lwf_s$
broadcasts page $p$ a new event occurs. An event $x$ for page $p$ will be defined as $E_{p,x} = \langle b_{p, x},
e_{p,x} \rangle$ where $b_{p,x}$ is the beginning of the event and $e_{p,x}$ is the end. Here $\lwf_s$ broadcast page
$p$ at time $b_{p,x}$ and this is the $x$th broadcast of page $p$. Then $\lwf_s$ broadcast page $p$ at time $e_{p,x}$
and this is the $(x+1)$st broadcast of page $p$. This starts a new event $E_{p,x+1}$. Therefore, the algorithm $\lwf_s$
does not broadcast $p$ on the time interval $[b_{p,x}+1, e_{p,x}-1]$. Thus, it can be seen that for page $p$, $e_{p,
x-1} = b_{p,x}$. It is important to note that the optimal offline solution may broadcast page $p$ multiple (or zero)
times during an event for page $p$.  See Figure~\ref{fig:event}.
\begin{figure}[tbh]
\begin{center}
\includegraphics[scale=.6]{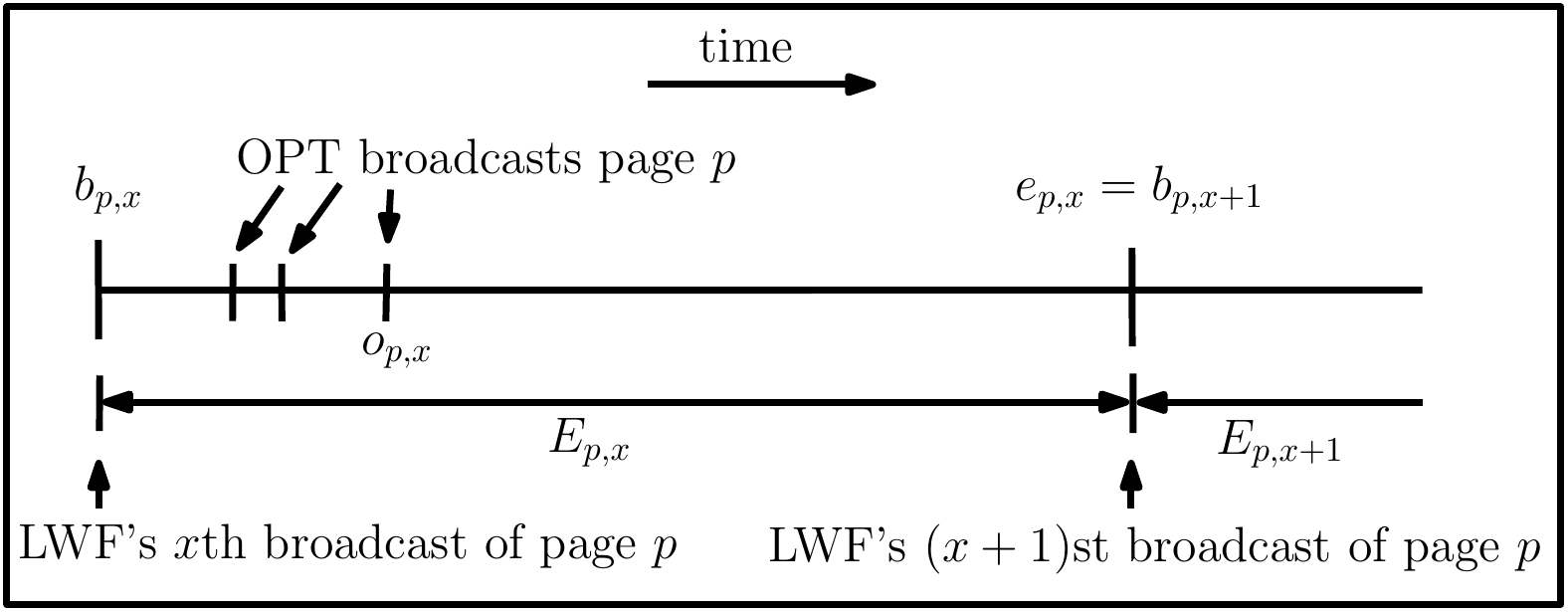}\vspace{-4mm}
\caption{Events for page $p$.} \label{fig:event}
\end{center}
\vspace{-9mm}
\end{figure}

For each event $E_{p,x}$ we let $\cJ_{p,x} = \{(p,i) \mid a_{p,i} \in [b_{p,x}, e_{p,x}-1]\}$ denote the set of
requests for $p$ that arrive in the interval $[b_{p,x}, e_{p,x}-1]$ and are satisfied by $\lwf_s$ at $e_{p,x}$. We let
$F_{p,x}$ denote the flow-time in $\lwf_s$ of all requests in $\cJ_{p,x}$. Similarly we define $F^*_{p,x}$ to be flow
time in $\opt$ for all requests in $\cJ_{p,x}$. Note that $\opt$ may or may not satisfy requests in $\cJ_{p,x}$ during
the interval $[b_{p,x}, e_{p,x}]$.

An event $E_{p,x}$ is said to be self-chargeable and in the set $S$ if
$F_{p,x} \leq F^*_{p,x}$ or $e_{p,x} - b_{p,x} < \len$, where $\len >
1$ is a constant which will be fixed later.  Otherwise the event is
non-self-chargeable and is in the set $N$. Implicitly we are
classifying the requests as self-chargeable or non-self-chargeable,
however it is easier to work with events rather than individual
requests. As the names suggest, self-chargeable events can be easily
charged to the flow-time of an optimal schedule. To help analyze the
flow-time for non-chargeable events, we set up additional notation and
further refine the requests in $N$.

Consider a non-self-chargeable event $E_{p,x}$.  Note that since this
event is non-self-chargeable, the optimal solution must broadcast page
$p$ during the interval $[b_{p,x}+1, e_{p,x}-1]$; otherwise, $F_{p,x}
\le F^*_{p,x}$ and the event is self-chargeable.  Let $o_{p,x}$ be the
last broadcast of page $p$ by the optimal solution during the interval
$[b_{p,x}+1, e_{p,x}-1]$. We define $o'_{p,x}$ for a
non-self-chargeable event $E_{p,x}$ as $\min\{o_{p,x}, e_{p,x} -
\len\}$. This ensures that the interval $[o'_{p,x},e_{p,x}]$ is
sufficiently long; this is for technical reasons and the reader should
think of $o'_{p,x}$ as essentially the same as $o_{p,x}$.

Let $\lwf^{S}_s = \sum_{p,x:E_{p,x} \in S} F_{p,x}$ and $\lwf_s^{N} =
\sum_{p,x: E_{p,x} \in N} F_{p,x}$ denote the the total flow time for
self-chargeable and non self-chargeable events
respectively. Similarly, let $\opt^{S} = \sum_{p,x: E_{p,x} \in S}
F^*_{p,x}$ and $\opt^{N} = \sum_{p,x: E_{p,x} \in N} F^*_{p,x}$.  For
a non-chargeable event $E_{p,x}$ we divide $\cJ_{p,x}$ into early
requests and late requests depending on whether the request arrives
before $o'_{p,x}$ or not.  Letting $\fe_{p,x}$ and $\fl_{p,x}$ denote
the flow-time of early and late requests respectively, we have
$F_{p,x} = \fe_{p,x} + \fl_{p,x}$. Let $\lwfne_s$ and $\lwfnl_s$
denote the total flow time of early and late requests of
non-self-chargeable events for $\lwf$'s schedule,
respectively.

The following two lemmas follow easily from the definitions.

\begin{lemma} \label{lem:SC} $\lwf_s^{S} \leq \len \opt^{S}$.
\vspace{-1mm}
\end{lemma}

\begin{lemma}
\label{lem:NSCl}
 $\lwfnl_s \leq \len  \optn$.
\end{lemma}

Thus the main task is to bound $\lwfne_s$. For a non-chargeable event
$E_{p,x}$ we try to charge $\fe_{p,x}$ to events ending in the interval
$[o'_{p,x}, e_{p,x}-1]$. The lemma below quantifies the relationship
between $\fe_{p,x}$ and the flow-time of events ending in this interval.

\begin{lemma}
  \label{lem:close} For any $0 \leq \lambda \leq 1$, if $e_{q,y} \in
  [\ceil{o'_{p,x} + \lambda (e_{p,x} - o'_{p,x})}, e_{p,x}-1]$ then $F_{q,y}
  \geq \lambda \fe_{p,x}$.
\end{lemma}

\begin{proof}
  Let $F_{p,x}(t)$ be the total waiting time accumulated by $\lwf$ for
  page $p$ on the time interval $[b_{p,x}, t]$. We divide
  $F_{p,x}(t)$ into two parts $F^e_{p,x}(t)$ and $F^l_{p,x}(t)$, which
  are the flow time due to early requests and to late requests,
  respectively. Note that $F_{p,x}(t) = F^e_{p,x}(t) + F^l_{p,x}(t)$.
  The early requests arrived before time $o'_{p,x}$, thus, for any $t'
  \ge \ceil{o'_{p,x} + \lambda (e_{p,x} - o'_{p,x})}$, $F^e_{p,x}(t')
  \geq \lambda \fe_{p,x}(e_{p,x}) = \lambda \fe_{p,x}$.

  Since $\lwf_s$ chose to transmit $q$ at $e_{q,y}$ when $p$ was
  available to be transmitted, it must be the case that $ F_{q,y} \ge
  F_{p,x}(e_{q, y}) \ge \fe_{p,x}(e_{q, y})$. Combining this with the
  fact that $\fe_{p,x}(e_{q, y}) \ge \lambda \fe_{p,x}$, the lemma follows.

\end{proof}

With the above setup in place, we now prove that $\lwf_s$ is $O(1)$
competitive for $s=5$ via a clean and simple proof, and for $s=4$ via
a more involved proof. These proofs can be extended to non-integer
speeds with better bounds on the speed. In particular, we can show
that $\lwf_{3.4 + \epsilon}$ is $O(1/\eps^3)$-competitive. We omit
these extensions in this version.

\vspace{-3mm}
\subsection{Analysis of $5$-Speed}\vspace{-2mm}
\label{sec:5speed}

This section will be devoted to proving the following main lemma that
bounds the flow-time of early requests of non self-chargeable events.

\begin{lemma}
  \label{lem:5speedNSC} For $\len \geq 1$, $\lwf_5^{N^e} \leq \frac{4
    \len}{5(\len-1)} \lwf_{5}$.
\end{lemma}

Assuming the lemma, $\lwf_5$ is $O(1)$-competitive, using the argument
outlined earlier in Section~\ref{sec:overview}.
\begin{theorem}
  \label{thm:5spd}
    $\lwf_{5} \leq 90 \opt$.
\end{theorem}

\begin{proof}

By combining Lemma~\ref{lem:SC}, \ref{lem:NSCl} and
\ref{lem:5speedNSC}, we have that $\lwf_{5} = \lwfs_{5} + \lwfnl_{5}
+\lwfne_{5} \leq \len \opts + \len \optn + \frac{4 \len}{5(\len-1)}
\lwf_{5}$. Setting $\len = 10$ completes the proof.
\end{proof}

We now prove Lemma~\ref{lem:5speedNSC}. In the analysis, we assume that $\lwf$ broadcasts $5$ pages at each time;
otherwise we can apply the same argument to maximal subintervals when $\lwf$ is fully busy, respectively. 
Let $E_{p,x} \in N$. We define
two intervals $I_{p,x} = [o'_{p,x}, e_{p,x}-1]$ and $I'_{p,x} =
[o'_{p,x}+\ceil{(e_{p,x} - o'_{p,x})/2}, e_{p,x}-1]$.  Since $\rho \le
e_{p,x} - o'_{p,x}$, it follows that $|I'_{p,x}| \ge
\frac{\len-1}{2\len} |I_{p,x}|$.  We wish to charge $\fe_{p,x}$ to
events (could be in $S$ or $N$) in the interval $I'_{p,x}$. By
Lemma~\ref{lem:close}, each event $E_{q,y}$ that finishes in
$I'_{p,x}$ satisfies the property that $F_{q,y} \ge
\fe_{p,x}/2$. Moreover, there are $5(\floor{e_{p,x} - o'_{p,x})/2}$
such events to charge to since $\lwf_5$ transmits $5$ pages in each
time slot.  Thus, locally for $E_{p,x}$ there are enough events to
charge to if $\len$ is a sufficiently large constant.  However, an
event $E_{q,y}$ with $e_{q,y} \in I'_{p,x}$ may also be charged by
many other events if we follow this simple local charging scheme. To
overcome this overcharging, we resort to a global charging scheme by
setting up a bipartite graph $G$ and invoking the fractional Hall's
theorem (see Lemma~\ref{lem:Hall_thm}) on this graph.

The bipartite graph $G=(X \cup Y, E)$ is defined as follows.  There is
exactly one vertex $u_{p,x} \in X$ for each non-self-chargeable event
$E_{p,x} \in N$ and there is exactly one vertex $v_{q,y} \in Y$ for
each event $E_{q,y} \in A$, where $A$ is the set of all events. Consider two vertices $u_{p,x} \in X$ and
$v_{q,y} \in Y$. There is an edge $u_{p,x}v_{q,y} \in E$ if and only
if $e_{q,y} \in I'_{p,x}$.  By Lemma~\ref{lem:close}, if there is an
edge between $u_{p,x} \in X$ and $v_{q,y} \in Y$ then $ F_{q,y} \ge
\fe_{p,x}/2$.

The goal is now to show that $G$ has a $\frac{2
  \len}{5(\len-1)}$-covering. Consider any non-empty set $Z \subseteq
X$ and a vertex $u_{p,x} \in Z$. Recall that the interval $I_{p,x}$
contains at least one broadcast by $\opt$ of page $p$. Let
$\mathcal{I} = \bigcup_{u_{p,x} \in Z } I_{p,x}$ be the union of the
time intervals corresponding to events in $Z$. Similarly, define
$\mathcal{I}' = \bigcup_{u_{p,x} \in Z } I'_{p,x}$.

We claim that $|Z| \leq |\mathcal{I}|$. This is because the optimal
solution has 1-speed and it has to do a separate broadcast for each
event in $Z$ during $\mathcal{I}$. Now consider the neighborhood of
$Z$, $N_G(Z)$.  We note that $|N_G(Z)| = 5 |\mathcal{I}'|$ since
$\lwf_5$ broadcasts $5$ pages for each time slot in $|\mathcal{I}'|$
and each such broadcast is adjacent to an event in $Z$ from the
definition of $G$.  From Lemma~\ref{lem:intervals}, $|\mathcal{I}'|
\ge \frac{\len-1}{2\len} |\mathcal{I}|$ as we had already observed
that $|I'_{p,x}| \ge \frac{\len-1}{2\len}|I_{p,x}|$ for each $E_{p,x}
\in N$. Thus we conclude that $|N_G(Z)| = 5 |\mathcal{I}'| \ge
5\frac{\len-1}{2\len} |\mathcal{I}| \ge 5\frac{\len-1}{2\len}
|Z|$. Since this holds for $\forall Z \subseteq X$, by
Lemma~\ref{lem:Hall_thm}, there must exist a $\frac{2 \len }{5(\len
  -1)}$-covering. Let $\ell$ be such a covering. Finally, we prove
that the covering implies the desired bound on $\lwfne_{5}$. \vspace{-2.5mm}

\begin{eqnarray*}
\lwfne_{5}\hspace{-1mm}
&=& \hspace{-0.6mm}\sum_{u_{p,x} \in X} \hspace{-1mm}\fe_{p,x} \mbox{ [By Definition]} \\
&=& \hspace{-2.8mm} \sum_{u_{p,x}v_{q,y} \in E} \hspace{-3.5mm} \ell_{u_{p,x},v_{q,y}} \fe_{p,x}  \mbox{ [By Def.~\ref{def:covering}, i.e. for \hspace{-0.5mm}$\forall u_{p,x} \hspace{-1mm}\in \hspace{-1mm}X$, $\sum_{v_{q,y \in Y}} \ell_{u_{p,x},v_{q,y}}\hspace{-0.9mm} = \hspace{-0.9mm}1$]}  \\
&\leq&  \hspace{-2.8mm} \sum_{u_{p,x}v_{q,y} \in E} \hspace{-3.5mm}\ell_{u_{p,x},v_{q,y}} 2 F_{q, y} \mbox{ [By Lemma~\ref{lem:close}]} \\
&\leq& \frac{4 \len}{5(\len-1)} \sum_{v_{q,y} \in Y} F_{q, y} \mbox{ [Change order of $\sum$ and $\ell$ is a $\frac{2
\len}{5(\len-1)}$-covering]} \\ \vspace{-1.5mm}
&\leq& \frac{4 \len}{5(\len-1)} \lwf_{5}. \mbox{ [Since $Y$ includes all events]} \\
\end{eqnarray*}

This finishes the proof of Lemma~\ref{lem:5speedNSC}.

\begin{remark}
  If non-integer speeds are allowed then the analysis in this
  subsection can be extended to show that $\lwf$ is ${4+ \eps}$-speed
  $O(1+1/\eps^2)$-competitive.
\vspace{-2.5mm}
\end{remark}

\subsection{Analysis of $4$-Speed}\vspace{-2mm}
\label{sec:4speed}
Due to insufficient space, we only sketch the key idea.  We remind the
reader that early requests of each non-self-chargeable event $E_{p,x}$
were charged to only half the events that ended on $[o'_{p,x},
e_{p,x}-1]$.  Thus, fully utilizing all the events, which end during
$[o'_{p,x}, e_{p,x}-1]$, can improve the speed.
Lemma~\ref{lem:close}, however, does not provide a good comparison
between $\fe_{p,x}$ and flow time of event $E_{r,z}$ which is done
close to $o'_{p,x}$. We overcome this by further refining the class of
non self-chargeable events into $\one$ and $\two$. For an event
$E_{p,x}$ in the interesting class $\two$, we are able to show that
all events in $[o'_{p,x}, e_{p,x}-1]$ have comparable flow-time to
that of $E_{p,x}$. This allows us to effectively charge $E_{p,x}$ to
events done at $o_{p,x}$; note that for any two events $E_{p,x}$ and
$E_{q,y}$ in $N$, $o_{p,x} \neq o_{q,y}$. The proof is technical and
requires several parameters; details can be found in
Appendix~\ref{sec:4speed}.

\vspace{-3mm}

\section{Generalization to Delay-Factor and  $L_k$ Norms}\vspace{-1mm}
\label{sec:metrics} \vspace{-2mm} In this section, our proof
techniques are extended to show that a generalization of $\lwf$ is
$O(1)$-speed $O(1)$-competitive for minimizing the average
delay-factor and minimizing the $L_k$-norm of the delay-factor.
Recall that flow-time can be subsumed as a special case of
delay-factor. Thus, these results will also apply to $L_k$ norms of
flow-time. Instead of focusing on specific objective functions, we
develop a general framework and derive results for delay-factor and
$L_k$ norms as special cases. First, we set up some notation. We
assume that for each request $J_{p,i}$ there is a non-decreasing
function $m_{p,i}(t)$ that gives the cost/penalty of that $J_{p,i}$
accumulates if it has {\em waited} for a time of $t$ units after its
arrival. Thus the total cost/penalty incurred for a schedule that
finishes $J_{p,i}$ at $f_{p,i}$ is $m_{p,i}(f_{p,i} -a_{p,i})$. For
flow-time $m_{p,i}(t) = t$ while for delay-factor it is $\max (1,
\frac{t-a_{p,i}}{d_{p,i} - a_{p,i}})$. For $L_k$ norms of delay-factor
we set $m_{p,i}(t) = \max (1, \frac{t-a_{p,i}}{d_{p,i} - a_{p,i}})^k$.
Note that the $L_k$ norm of delay-factor for a given sequence of
requests is $\sqrt[k]{\sum_{p,i} m_{p,i}(f_{p,i}-a_{p,i})}$ but we can
ignore the outer $k$'th root by focusing on the inner sum.

A natural generalization of $\lwf$ to more general metrics is described below; we refer to this (greedy) algorithm as
$\grdy$ for Longest First. We in fact describe $\grdy_s$ which is given $s$ speed over the adversary. \vspace{-2mm}
\begin{center}
\begin{tabular}[r]{|c|}
\hline
\textbf{Algorithm}: $\grdy_s$ \\
\begin{minipage}{\textwidth}
\begin{itemize}
\item At any integer time $t$, broadcast the $s$ pages with the
largest $m$-waiting times where the $m$-waiting time of page $p$ at $t$
 is  $\sum_{J_{p,i} \in U(t)} m_{p,i}(t -a_{p,i})$.
\end{itemize}
\end{minipage}\\
\hline
\end{tabular}
\end{center}

\begin{remark}
  The algorithm and analysis do not assume that the functions $m_{p,i}$
  are ``uniform'' over requests.  In principle each request $J_{p,i}$
  could have a different penalty function.
\end{remark}
In order to analyze $\grdy$, we need a lower bound on the ``growth''
rate of the functions $m_{p,i}()$. In particular we assume that there
is a function $h: [0,1] \rightarrow \mathbb{R}^+$ such that
$m_{p,i}(\lambda t) \ge h(\lambda) m_{p,i}(t)$ for all $\lambda \in
[0,1]$. It is not to difficult to see that for flow-time and delay-factor
we can choose $h(\lambda) = \lambda$, and for $L_k$ norms of
flow-time and delay-factor, we can set $h(\lambda) = \lambda^k$.
We also define a function $m:  \mathbb{R}^+ \rightarrow \mathbb{R}^+$
as $m(x) = \max_{(p,i)} m_{p,i}(x)$. The rest of the analysis depends
only on $h$ and $m$.

In the following subsection we outline a generalization of the
analysis from Section~\ref{sec:5speed} that applies to $\grdy_s$; the
analysis bounds various quantities in terms of the functions $h()$ and
$m()$.  In Section~\ref{sec:metric-results}, we derive the results for
minimizing delay-factor and $L_k$ norms of delay-factor.
\vspace{-2.5mm}

\subsection{Outline of Analysis} \vspace{-2mm}
\label{sec:gen-framework} To bound the competitiveness of $\grdy_s$, we use the same techniques we used for bounding
the competitiveness of $\lwf_s$. Events are again defined in the same fashion; $E_{p,x}$ is the event defined by the
$x$'th transmission of $p$ by $\grdy_s$. We again partition events into self-chargeable and non self-chargeable events
and charge self-chargeable events to the optimal value and charge non-self-chargeable events to $\delta \grdy_s +
m(\len) \optn$ for some $\delta < 1$. For an event $E_{p,x}$, let $M_{p,x}(t) = \sum_{J_{p,i} \in \mathcal{J}_{p,x}}
m_{p,i}(t-a_{p,i})$ denote the total $m$-cost of all requests for $p$ that arrive in $[b_{p,x}, e_{p,x}-1]$ that are
satisfied at $e_{p,x}$. We let $M^*_{p,x}(t)$ be the $m$-cost of the same set of requests for the optimal solution. An
event $E_{p,x}$ is self-chargeable if $M_{p,x} \le m(\len) M^*_{p,x}$ or $e_{p,x} - b_{p,x} \le \len$ for some constant
$\len$ to be optimized later. The remaining events are non self-chargeable.  Again, requests for non-self-chargeable
events are divided into early requests and late requests based on whether they arrive before $o'_{p,x}$ or not where
$o'_{p,x} = \min\{o_{p,x}, e_{p,x}-\len\}$.  Let $M^e_{p,x}$ and $M^l_{p,x}$ be the flow time accumulated for early and
late requests of a non-self-chargeable event $E_{p,x}$, respectively.  The values of $\grdy_s^{N}$, $\grdy_s^{N^l}$,
$\grdy_s^{N^e}$, and $\grdy_s^{S}$ are defined in the same way as $\lwf_s^{N}$, $\lwf_s^{N^l}$, $\lwf_s^{N^e}$, and
$\lwf_s^{S}$. Likewise for $\opt$.  The following two lemmas are analogues of Lemmas~\ref{lem:SC} and \ref{lem:NSCl}
and follow from definitions.

\begin{lemma} \label{lem:GSC}
  $\grdy_s^{S} \leq m(\len) \opt^{S}$.
\end{lemma}

\begin{lemma}
  \label{lem:GNSCl}
  $\grdy^{N^l}_s \leq m(\len) \opt^{N}$.
\end{lemma}

We now show a generalization of Lemma~\ref{lem:close} that states that
any event $E_{q,y}$ such that $e_{q,y}$ is close to $e_{p,x}$ has
$m$-waiting time comparable to the $m$-waiting time of early requests of
$E_{p,x}$.

\begin{lemma}
  \label{lem:rq_close} Suppose $E_{p,x}$ and $E_{q,y}$ are two events
  such that $e_{q,y} \in [\ceil{o'_{p,x} + \lambda(e_{p,x} -
    o'_{p,x})}, e_{p,x}-1]$, $M_{q,y} \geq h(\lambda)M^e_{p,x}$.
\end{lemma}
\begin{proof}[Sketch]
  Consider an early request $J_{p,i}$ in $\cJ_{p,x}$ and let $t \in
  [\ceil{o'_{p,x} + \lambda(e_{p,x} - o'_{p,x})}, e_{p,x}-1]$.  Since
  $a_{p,i} \le o'_{p,x}$, it follows that $t \ge \lambda (e_{p,x} -
  a_{p,i}) + a_{p,i}$. Hence, $m_{p,i}(t-a_{p,i}) \ge h(\lambda)m_{p,i}(e_{p,x}
  - a_{p,i})$. Summing over all early requests, it follows that
  $M^e_{p,x}(t) \ge h(\lambda) M^e_{p,x}$. Since $\grdy_s$ chose
  to transmit $q$ at $t = e_{q,y}$ instead of $p$, it follows that
  $M_{q,y} \ge M_{p,x}(e_{q,y}) \ge M^e_{p,x}(e_{q,y}) \ge h(\lambda)M^e_{p,x}$.
\end{proof}

As in Section~\ref{sec:5speed}, the key ingredient of the analysis is
to bound the waiting time of early requests. We state the analogue of
Lemma~\ref{lem:5speedNSC} below. Observe that we have an additional
parameter $\beta$. In Lemma~\ref{lem:5speedNSC} we hard wire $\beta$ to
be $1/2$ to simplify the exposition. In the more general setting, the
parameter $\beta$ needs to be tuned based on $h$.

\begin{lemma}
    \label{lem:GNSC}
    For any $0<\beta<1$, $\grdy_s^{N^e} \leq \frac{ \len
    }{s h(\beta)(\len(1-\beta)  -1)} \grdy_{s} $, where $h$ is some
    scaling function for $m$.
\end{lemma}

The proof of the above lemma follows essentially the same lines as
that of Lemma~\ref{lem:5speedNSC}. The idea is to charge $M^e_{p,x}$
to events in the interval $[o'_{p,x} + \ceil{\beta(e_{p,x} -
  o'_{p,x})}, e_{p,x}-1]$. Using Lemma~\ref{lem:rq_close}, each
event in this interval is within a factor of $h(\lambda)$ of $M^e_{p,x}$.  The length of this interval is at least
$\frac{\len(1-\beta) -1}{\rho}$ times the length of the interval $[o'_{p,x}, e_{p,x}-1]$. To avoid overcharging we
again resort to the global scheme using fractional Hall's theorem after we setup the bipartite graph. We can then prove
the existence of a $\frac{ \len }{s (\len(1-\beta) -1)}$-covering and since each event can pay to within a factor of
$h(\beta)$, the lemma follows.

Putting the above lemmas together we derive the following theorem.

\begin{theorem}
    \label{thm:G}
    Let $\beta \in (0,1)$ and $\len > 1$ be given constants.  If $s$
    is an integer such that $\frac{ \len }{sh(\beta)(\len(1-\beta) -1)} \le \delta < 1$, then algorithm $\grdy_s$ is $s$-speed
    $\frac{m(\len)}{1-\delta}$-competitive.
\vspace{-2.5mm}
\end{theorem}

\subsection{Results for Delay-Factor and $L_k$ Norms}
\label{sec:metric-results} \vspace{-0.5mm}

We can apply Theorem~\ref{thm:G} with appropriate choice of parameters
to show that $\grdy_s$ is $O(1)$-competitive with $O(1)$ speed.

For minimizing average delay-factor we have $h(\lambda) = \lambda$ and
$m(x) \le x$. For this reason, average delay-factor behaves essentially
the same as average flow-time and we can carry over the results
from flow-time.

\begin{theorem}
  \label{thm:df_34spd}
  The algorithm $\grdy$ is $5$-speed $O(1)$
  competitive for minimizing the average delay-factor. For non-integer speeds
  it is $4+\eps$-speed $O(1/\eps^2)$-competitive.
\end{theorem}

The analysis in Section~\ref{sec:4speed} also extends to delay-factor although it does not fall in the general
framework that we outlined in Section~\ref{sec:gen-framework}. Thus $\grdy$ is $(3.4+\eps)$-speed
$O(1/\eps^3)$-competitive for average delay-factor.

For $L_k$ norms of delay-factor we have $h(\lambda) = \lambda^k$ and $m(x) \le x^k$. By choosing $\beta =
\frac{k}{k+1}$, $\rho = 90(k+1)$ and $s = 3(k+1)$ in Theorem~\ref{thm:G}, we can show that
the algorithm $\grdy$ is $3(k+1)$-speed $O(\rho^k)$-competitive for minimizing $\sum_{p,i} m_{p,i}(f_{p,i}-a_{p,i})$.
Thus for minimizing the $L^k$-norm delay factor, we obtain  $\sqrt[k]{O(\rho^k)} = O(\rho)$ competitiveness, which shows the following.

\begin{theorem}
    \label{thm:LP}
    For $k \ge 1$, the algorithm $\grdy$ is $O(k)$-speed
    $O(k)$-competitive for minimizing $L_k$-norm of delay-factor.
\end{theorem}

\vspace{-7mm}

\section{Conclusion}\vspace{-3mm}
\label{sec:concl}
We gave a simpler analysis of $\lwf$ for minimizing average flow-time
in broadcast scheduling. This not only helps improve the speed bound
but also results in extending the algorithm and analysis to more
general objective functions such a delay-factor and $L_k$ norms of
delay-factor. We hope that our analysis is useful in other scheduling
contexts.

Our recent work in \cite{ChekuriIM09} shows that $\grdy$ is not
$O(1)$-competitive with any speed for $L_\infty$-norm of delay factor,
which is equivalent to minimizing the maximum delay factor. Thus, we
believe the speed requirement for $\grdy$ to obtain
$O(1)$-competitiveness needs to grow with $k$ for $L_k$-norms of delay
factor. It would be interesting to formally prove this. This raises
the question of whether there is an alternate algorithm that is
$O(1)$-speed $O(1)$-competitive for $L_k$ norms of flow time and delay
factor. We remark that the lower bound for $\grdy$ \cite{ChekuriIM09}
applies only to delay factor and it is open whether $\grdy$ is
$O(1)$-speed $O(1)$-competitive for $L_k$ norms of flow time.
It would be also an interesting direction
to find an $\lwf$-like algorithm that performs well when page sizes are different.

As we mentioned earlier, we have very recently obtained a tighter bound on $\lwf$: 
for any $\eps >0$, it is $O(1)$-competitive with $(2.777+ \eps)$-speed and is not $O(1)$-competitive 
with $(2 - \eps)$-speed. We conjecture that our new lowerbound is tight. 
We will include our new results in the full version.

\bigskip
\noindent {\bf Acknowledgments:} We thank Kirk Pruhs for his helpful
comments and encouragement.

\bibliographystyle{plain}
\bibliography{LWF2}
\appendix


\section{Analysis of 4-speed}

\label{sec:4speed} In this section, we further improve the speed from $5$ to $4$ in the discrete time model. We assume
the speed $s = 4$ throughout this section.

\begin{theorem}
\label{thm:4spd} $\lwf$ is $4$-speed $O(1)$-competitive.
\end{theorem}
\begin{proof}
    By combining Lemma~\ref{lem:SC}, \ref{lem:NSCl}, \ref{lem:type1} and \ref{lem:type2} (Lemma~\ref{lem:type1} and \ref{lem:type2} will be proved soon), it follows that
\begin{eqnarray*}
                 \lwf_{4}   &=   & \lwfs_4 + \lwfnl_4 + \lwfneone_4 + \lwfnetwo_4 \\
                            &\leq& \len\opts + \len\optn + \frac{4 \len}{ \alpha^2} \opt + \frac{3 - 8\alpha - 8\gamma}{4(1 - 4\alpha - 4\gamma)} \lwf_4 + \frac{\len}{4\gamma} \optn
\end{eqnarray*}
Setting $\len = 128$, $\alpha = 1/32$ and $\gamma = 1/32$ completes the proof.
\end{proof}

In Section~\ref{sec:5speed} early requests of each non-self-chargeable event $E_{p,x}$ were charged to events that
ended on the last half of $[o'_{p,x}, e_{p,x}-1]$. This was a compromise between using more events vs. finding quality
events. In other words, if we use more events ending in $[o'_{p,x}, e_{p,x}-1]$, the average quality of those events degrades because  events ending close to $o'_{p,x}$ do not have flow time comparable to $\fe_{p,x}$. On the other hand,
if we use only quality events, we can only charge to a small faction of events ending on $[o'_{p,x}, e_{p,x}-1]$. To overcome this issue, we will show that all events
ending in $[o'_{p,x}, e_{p,x}-1]$ have comparable flow time with $\fe_{p,x}$ if only a small number of self-chargeable events end on $[o'_{p,x}, e_{p,x}-1]$.  This will then improve our bound on the speed.   For the other case where $E_{p,x}$ has
many self-chargeable events, $\fe_{p,x}$ will be directly charged to those self-chargeable events having comparable
flow time with $\fe_{p,x}$, thus directly to $\opt$.

We now describe this idea in more details. Non-self-chargeable events in $N$ are partitioned into two disjoint sets
$N_1$ and $N_2$ depending on they have many self-chargeable events or not. Formally, Non-self-chargeable event
$E_{p,x}$ is said to be $\one$ and in $N_1$ if it has at least $\alpha s (e_{p,x}- o'_{p,x})$ self-chargeable events
where $\alpha <1$ is some constant to be fixed later. The rest of the events in $N$ are in $N_2$ and said to be $\two$.
We let $\lwf^{N_1e}_4$ and $\lwf^{N_2e}_4$ denote the total flow time of early requests of $N_1$ and $N_2$, respectively.
As mentioned already, the $\one$ events can be charged to the optimal solution because it has many self-chargeable
events. For each $\two$ event, we will bound $\fe_{p,x}$ with events which end at $o_{p,x}$. Recall that
Lemma~\ref{lem:close} cannot compare $\fe_{p,x}$ and $F_{r,z}$, where $E_{r,z}$ is an event ending at $o_{p,x}$, i.e.
$e_{r,z} = o_{p,x}$. Thus we find a \emph{bridge} events which start from a way before $o'_{p,x}$ and end close to
$e_{p,x}$. Since each bridge event $E_{q,y}$ substantially overlap both with $E_{p,x}$ and with $e_{r,z}$, we can
compare $\fe_{p,x}$ with $F_{q,y}$ and $\fe_{q,y}$ with $F_{r,z}$, thereby $\fe_{p,x}$ with $F_{r,z}$. We also observe
that each $E_{r,z}$ is charged by one event $E_{p,x}$ such that $o_{p,x} = e_{r,z}$, as each non-self-chargeable event
has its unique last broad cast time. Thus we are safe from overcharging.

In the following lemma, we directly charge early requests of $\one$ events to $\opt$. For the goal, we charge early
requests of each $\one$ event $E_{p,x}$ to self-chargeable events having flow time comparable to $\fe_{p,x}$ which end
on $[o'_{p,x}, e_{p,x} - 1]$. By the definition of $\one$ events, we know that each $\one$ event $E_{p,x}$ has many
events that it can be charged to. However, to the overcharging issue, we resort to a global charging scheme again using
the modified Hall's theorem. We separate how to find a covering to Lemma~\ref{lem:cover2}, as we will use it again for
charging $\two$ events.

\begin{lemma}
\label{lem:type1} If $\alpha \len \geq 4$, then $\lwf^{N_1e}_4 \leq \frac{4 \len}{ \alpha^2} \opt$.
\end{lemma}
\begin{proof}
  Let $G=(X \cup Y,E)$ be a bipartite graph where $u_{p,x} \in
  X$ iff $E_{p,x} \in N_1$, $v_{q,y} \in Y$ iff $E_{q,y} \in S$, and
  $u_{p,x}v_{q,y} \in E$ iff $e_{q,y} \in [o'_{p,x} + \lceil \alpha/2
  (e_{p,x} - o'_{p,x}) \rceil, e_{p,x}-1]$.
Note that if $u_{p,x}v_{q,y} \in E$, $\fe_{p,x} \leq \frac{2}{\alpha}F_{q,y}$ by Lemma~\ref{lem:close}.
  It can be observed that each vertex
  $u_{p,x} \in X$ has at least $2\alpha (e_{p,x} -
  o'_{p,x})-4 \,( \geq \alpha (e_{p,x} - o'_{p,x})\,\,\textnormal{by the given condition})$ neighbors.
  This follows from the observations that at least $\alpha s(e_{p,x} - o'_{p,x})$ self-chargeable
  events end during $[o'_{p,x}, e_{p,x}-1]$ by definition of $\one$ and
  at most $s(\alpha/2(e_{p,x} - o'_{p,x}))+s$ events end
  during $[o'_{p,x}, o'_{p,x} + \lceil \alpha/2 (e_{p,x} - o'_{p,x}) \rceil -1]$. Thus
  $G$ has a $\frac{2}{\alpha}$-covering by Lemma~\ref{lem:cover2}. Let
  $\ell$ be such a covering. We now prove the final step.
\begin{eqnarray*}
\lwfneone_{4}
&=   &  \sum_{u_{p,x} \in X} \fe_{p,x} = \sum_{u_{p,x}v_{q,y} \in E} \ell_{u_{p,x},v_{q,y}}  \fe_{p,x} \mbox{[By Definition~\ref{def:covering}]}  \\
&\leq&  \sum_{u_{p,x}v_{q,y} \in E} \ell_{u_{p,x},v_{q,y}} \frac{2}{\alpha}F_{q, y} \mbox{[By Lemma~\ref{lem:close}]}  \\
&\leq& \frac{2}{\alpha} \frac{2}{\alpha} \sum_{v_{q ,y} \in Y} F_{q, y}    \mbox{[Change order of summation and $\ell$ is $\frac{2}{\alpha}$-covering]}  \\
&\leq& \frac{4}{ \alpha^2 } \lwfs_4 \mbox{ [Since $Y$ includes all self-chargeable events]}  \\
&\leq& \frac{4}{ \alpha^2 } \len \opts \mbox{[By Lemma~\ref{lem:SC}]}
\end{eqnarray*}
\end{proof}

The following lemma states, when $G$ is a bipartite graph whose parts are a subset of non-self-chargeable events and a
subset of all events respectively, the quality of covering in terms of how many neighbors each non-self-chargeable
event has. The main difference from what was done for finding a covering in the proof of Lemma~\ref{lem:5speedNSC} is
that here each non-self-chargeable event is not required to have all events ending in some sub-interval as its
neighbors.

\begin{lemma}
    \label{lem:cover2}
    Let $A$ denote all events. Let $G = (X \cup Y, E)$ be a bipartite graph where there exists only one vertex $u_{p,x} \in X$ only if $E_{p,x} \in N$,
    there exists only one vertex $v_{q,y} \in Y$ only if $E_{q,y} \in A$ and $v_{q,y} \in N_G(u_{p,x})$ only if $e_{q,y} \in [o'_{p,x}, e_{p,x}-1]$.
     Suppose that $\exists \lambda>0$ such that $\forall u_{p,x} \in X,  |N_G(u_{p,x})| \geq \lambda(e_{p,x} - o'_{p,x})$. Then
     there exists $\frac{2}{\lambda}$-covering for $X$.
\end{lemma}
\begin{proof}
  Consider any non-empty set $Z \subseteq X$ and its neighborhood $N(Z)$.
  We will show that $|N_G(Z)| \geq \lambda/2 |G|$. Let $I_{p,x} = [o'_{p,x}, e_{p,x}-1]$ and $\mathcal{I} = \bigcup_{u_{p,x}
    \in Z } I_{p,x}$. For simplicity we assume that $\mathcal{I}$ is a
  contiguous interval. Otherwise, the proof can be simply reduced to
  each maximal contiguous interval in $\mathcal{I}$.  First we show
  $N_G(Z) \geq \frac{\lambda}{2} |\mathcal{I}|$. We generously give up
  intervals in $\mathcal{I}$ which are contained in other intervals in
  $\mathcal{I}$ and order the remaining intervals in increasing order
  of their starting points. After picking up the first interval, we
  greedily pick up the next interval which the least overlaps with
  the previous chosen interval or starts just after the end of the interval.
  We index the chosen intervals according
  to their orders, 1,2,3 and so on.  Let $\mathcal{I}_{odd}$ and
  $\mathcal{I}_{even}$ be the odd-indexed and even-indexed intervals,
  respectively. Note that no intervals in $\mathcal{I}_{odd}$ overlap with each other.
  Likewise for $\mathcal{I}_{even}$. We have $|\mathcal{I}_{even}| + |\mathcal{I}_{odd}|
  \geq |\mathcal{I}|$, since $\mathcal{I} = \mathcal{I}_{even} \cup
  \mathcal{I}_{odd}$. WLOG, suppose
  $|\mathcal{I}_{odd}| \geq |\mathcal{I}_{even}|$.
  Let us consider any interval $I_{p',x'}$ in $\mathcal{I}_{odd}$.
  We know that $E_{p',x'}$ (or $u_{p',x'}$) has at least $\lambda(e_{p,x} - o'_{p,x})$,
  so by summing over all intervals in $\mathcal{I}_{odd}$, we can find at least
  $\lambda |\mathcal{I}_{odd}| \geq \lambda/2 |\mathcal{I}|$. Thus we have $|N_G(Z)| \geq \lambda/2 |\mathcal{I}|$
Also we have $|Z| \leq |\mathcal{I}|$; this is because the optimal solution has 1-speed and since it has to do a
separate broadcast for each event in $Z$. Combining these two inequalities, if follows that $|N_G(Z)| \geq
\frac{\lambda}{2}|\mathcal{Z}|$, and therefore $G$ has $\frac{2}{\lambda}$-covering by Lemma~\ref{lem:Hall_thm}.
\end{proof}

Our attention is now shifted to $\two$ events. As mentioned previously, the main idea is to find bridge events for each
$E_{p,x} \in N_2$. Formally, $E_{q,y}$ is said to be a \emph{bridge} event of $E_{p,x}$ if $o'_{q,y} \leq e_{p,x} - (2
- 4\alpha - 4\gamma )(e_{p,x}-o'_{p,x})$ and $e_{q,y} \in [o'_{p,x} + \lceil 1/2(e_{p,x}-o'_{p,x}) \rceil, e_{p,x}-1]$,
where $0<\gamma<1$ is a constant to be decided later. Let $\mathcal{B}(E_{p,x})$ be the set of bridge events of
$E_{p,x}$. Recall that we want to compare $E_{p,x}$ with $E_{r,z}$ such that $e_{r,z} = o_{p,x}$. Intuitively, a bridge
event $E_{q,y}$ bridges two events $E_{p,x}$ and $E_{r,z}$ by stretching over both events.   The following lemma says
that every $\two$ event has many bridge events.

\begin{lemma}
  \label{lem:bdg_num} If $4 \gamma \len \geq 1$, then for any $E_{p,x}
  \in N_2$, $|\mathcal{B}(E_{p,x})| \geq 4 \gamma
  (e_{p,x}-o'_{p,x})\geq 1$.
\end{lemma}
\begin{proof}
  Let $E_{p,x} \in N_2$. Let $I = [o'_{p,x}, e_{p,x}-1]$ and $I' =
  [o'_{p,x} + \lceil 1/2(e_{p,x}-o'_{p,x}) \rceil, e_{p,x}-1]$.
  Our argument is simple; because there are many non-self-chargeable events ending in $I'$,
  the last optimal broadcast times of many of those events cannot be contained in $I'$, thus many events start a way
  earlier than $o'_{p,x}$. For the formal proof, we first show that
  (1)$\,$there are at least $(2 - 4 \alpha)(e_{p,x}-o'_{p,x}) -2$
  non-self-chargeable events that end during $I'$. This is because
  there are at least $s \lfloor 1/2(e_{p,x}-o'_{p,x}) \rfloor \geq 2(e_{p,x}-o'_{p,x})-2$ events which end during
  $I'$ and $\two$ event $E_{p,x}$ has at most $\alpha s
  (e_{p,x}-o'_{p,x})$ self-chargeable events which end during
  $I$ by definition. Note that for any non-self-chargeable event $E_{q,y}$ which
  ends on $I'$, $\opt$ must broadcast  page $q$ before $e_{p,x}$,
  more precisely $o_{q,z} < e_{q,z} < e_{p,x}$, that is $o_{q,z} \leq
  e_{p,x} -2$. Let $t_b = e_{p,x} - (2 - 4\alpha - 4\gamma
  )(e_{p,x}-o'_{p,x})$.  Finally, (2)$\,$there are at most $(2 -
  4\alpha - 4\gamma)(e_{p,x} - o'_{p,x}) - 2$ time slots when $\opt$
  can broadcast pages during $[\lceil t_b \rceil, e_{p,x}-2]$.
  From (1) and (2), we can deduce that $|\mathcal{B}(E_{p,x})| \geq 4 \gamma
  (e_{p,x}-o'_{p,x}) \geq 4 \gamma \len \geq 1$.

\end{proof}

In the next lemma, we show each bridge event $E_{q,y} \in \mathcal{B}(E_{p,x})$ provides a good comparison between
$\fe_{p,x}$ and the flow time of any event $F_{r,z}$ which end at $o_{p,x}$.

\begin{lemma}
  \label{lem:bdg} Suppose that $4 \gamma \len \geq 1$. Let $E_{p,x}
  \in N_2$, $E_{q,y} \in \mathcal{B}(E_{p,x})$ and $E_{r,z}$ be an
  event s.t. $e_{r,z} = o_{p,x}$.  Then, $\fe_{p,x} \leq \frac{3 -
    8\alpha - 8\gamma}{1 - 4\alpha - 4\gamma} F_{r,z} + 2 \len
  F^*_{q,y}$.
\end{lemma}
\begin{proof}
  We start from an easy case that $e_{r,z} \geq e_{q,y}$. We have
  $\frac{1}{2}\fe_{p,x} \leq \frac{ e_{r,z} - o'_{p,x}}{e_{p,x} -
    o'_{p,x}} \fe_{p,x} \leq F_{r,z}$. The first inequality comes from
  that $e_{q,y} \geq o'_{p,x} + \lceil \frac{1}{2}(e_{p,x} - o'_{p,x}) \rceil$ and
  the second by Lemma~\ref{lem:close}. Thus it holds that $\fe_{p,x}
  \leq 2 F_{r,z}$, which clearly satisfies the lemma.

  Now let us consider the other case that $e_{r,z} < e_{q,y}$. By
  comparing $E_{p,x}$ and $E_{q,y}$, using Lemma~\ref{lem:close}, we
  have (1) $\frac{1}{2}\fe_{p,x} \leq \frac{ e_{q,y} -
    o'_{p,x}}{e_{p,x} - o'_{p,x}} \fe_{p,x} \leq F_{q,y}$. The first
  inequality holds because $e_{q,y} \geq o'_{p,x} + \lceil \frac{1}{2}(e_{p,x} -
  o'_{p,x}) \rceil$ and the second by Lemma~\ref{lem:close}. Next we compare $E_{q,y}$ and $E_{r,z}$. It follows that
  (2) $\frac{2(1-4\alpha - 4\gamma)}{3 - 8\alpha - 8\gamma} \fe_{q,y}
  \leq \frac{ o'_{p,x} - o'_{q,y}}{e_{q,y} - o'_{q,y}} \fe_{q,y} \leq
  \frac{ e_{r,z} - o'_{q,y}}{e_{q,y} - o'_{q,y}} \fe_{q,y} \leq
  F_{r,z}$. The first inequality can be shown by easy calculation
  using the fact that $o'_{q,y} \leq e_{p,x} - (2 -4\alpha -
  4\gamma)(e_{p,x} - o'_{p,x})$ and $e_{q,y} \geq o'_{p,x} + \lceil 1/2(e_{p,x} - o'_{p,x}) \rceil$.
  The second follows from that $o'_{p,x} \leq o_{p,x} =
  e_{r,z}$. Combining (1) and (2), we get $\fe_{p,x} \leq 2F_{q,y} = 2
  (\fe_{q,y} + \fl_{q,y}) \leq \frac{3 - 8\alpha - 8\gamma}{1 -
    4\alpha - 4\gamma} F_{r,z} + 2 \len F^*_{q,y}$.
\end{proof}

\begin{remark}
Lemma~\ref{lem:bdg} holds for any event $E_{r,z}$ such that $e_{r,z} \in [o'_{p,x}, e_{p,x} -1]$. But we only need to
consider the case where $e_{r,z} = o_{p,x}$ for our charging scheme.
\end{remark}

By taking the average of the inequalities in Lemma~\ref{lem:bdg} over the $s=4$ events ending at $o_{p,x}$, we have the
following corollary.

\begin{corollary}
    \label{cor:bdg}
    Suppose that $4 \gamma \len \geq 1$. Let $E_{p,x} \in N_2$ and
    $E_{q,y} \in \mathcal{B}(E_{p,x})$. \\ Then, $\fe_{p,x} \leq \frac{3
      - 8\alpha - 8\gamma}{4(1 - 4\alpha - 4\gamma)} (\sum_{E_{r,z}| e_{r,z} =
      o_{p,x}} F_{r,z}) + \frac{\len}{2} F^*_{q,y}$
\end{corollary}

Note that in Lemma~\ref{lem:bdg}, $\fe_{p,x}$ is bounded not only with $F_{r,z}$ but also with $F^*_{q,y}$, which
contributes to $\opt$. If many events use $E_{q,y}$ as their bridges, $E_{q,y}$ can be overcharged. To avoid this, we
found many bridge candidates for each $\two$ event in Lemma~\ref{lem:bdg_num}. Using the modified Hall's theorem, we
will bound the number of events which use the same bridge event.

Now we are ready to bound early requests of $\two$ events, i.e. $\lwfnetwo$. Recall that each $\two$ event $E_{p,x}$ is
charged to the $s=4$ events which are finished at $o_{p,x}$. Note that $E_{r,z}$ is used only by $E_{p,x}$ since
$E_{p,x}$ is the only event such that $o_{p,x} = e_{r,z}$. Thus $E_{r,z}$ is not overcharged.

\begin{lemma}
    \label{lem:type2}
    If $4 \gamma \len \geq 1$, $\lwfnetwo_4 \leq \frac{3 - 8\alpha - 8\gamma}{4(1 - 4\alpha -
      4\gamma)} \lwf_4 + \frac{\len}{4\gamma} \optn$.
\end{lemma}
\begin{proof}
  Let $G=(X \cup Y,E)$ be a bipartite graph where $u_{p,x} \in
  X$ iff $E_{p,x} \in N_2$, $v_{q,y} \in Y$ iff $E_{q,y} \in N$ and
  $u_{p,x} v_{q,y} \in E$ iff $E_{q,y} \in \mathcal{B}(E_{p,x})$. By
  Lemma~\ref{lem:bdg_num}, $u_{p,x} \in X$ has at least $ 4\gamma
  (e_{p,x}-o'_{p,x})$ neighbors, hence by Lemma~\ref{lem:cover2}, $G$
  has $ \frac{1}{2\gamma}$-covering. Let $\ell'$ be such a covering.
  Now we are ready to prove the final step. For simplicity,
  let $k = \frac{3 - 8\alpha - 8\gamma}{4(1 - 4\alpha - 4\gamma)}$.
\begin{eqnarray*}
  \lwfnetwo_{4}
  &=   &  \sum_{u_{p,x} \in X} \fe_{p,x} = \sum_{u_{p,x}v_{q,y} \in E} \ell'_{u_{p,x},v_{q,y}} \fe_{p,x} \mbox{[By Definition~\ref{def:covering}]}  \\\\
  &\leq&  \sum_{u_{p,x}v_{q,y} \in E} \ell'_{u_{p,x},v_{q,y}} (k \sum_{E_{r,z} | e_{r,z} = o_{p,x}} F_{r,z} + \frac{\len}{2} F^*_{q,y})   \mbox{[By Corollary~\ref{cor:bdg}]}\\
  &=   &  k \sum_{u_{p,x} \in X} \sum_{E_{r,z}| e_{r,z} = o_{p,x}} F_{r,z} + \frac{\len}{2} \sum_{v_{q,y} \in Y} F^*_{q,y} \sum_{u_{p,x} \in X} \ell'_{u_{p,x},v_{q,y}} \\
  &\leq&  k \lwf_4 + \frac{\len}{2} \sum_{v_{q,y} \in Y} F^*_{q,y} \frac{1}{2\gamma}        \mbox{ [By (*) and $\ell'$ being a $\frac{1}{2\gamma}$-covering]}  \\
  &\leq&  k \lwf_4 + \frac{\len}{4\gamma} \optn      \mbox{ [Since $Y$ include all non-self-chargeable events]}  \\
\end{eqnarray*}
It holds that (*)~$\sum_{u_{p,x} \in X} \sum_{E_{r,z}|e_{r,z} = o_{p,x}} F_{r,z} \leq \lwf_4$, because for each
non-self-chargeable $E_{r,z}$ there is only one event $E_{p,x}$ such that $e_{r,z} = o_{p,x}$.
\end{proof}

\begin{remark}
  If non-integer speeds are allowed then the analysis in this
  subsection can be extended to show that $\lwf$ is ${3.4+ \eps}$-speed
  $O(1+1/\eps^3)$-competitive.
\end{remark}

\section{Omitted Proofs}
\subsection{Proof of Lemma~\ref{lem:intervals}}
\begin{proof}
Let $I$ be the union of all intervals in $X$. $I'$ is similarly defined for $X'$. We prove the lemma when $I'$ is a
contiguous interval; otherwise we can simply sum over all maximal intervals in $I'$. WLOG, we can set $I = [s_1, t']$
and $I'= [s', t']$. This is because $I$ must start with one interval in $X$, say $[s_1, t_1]$ and both $I$ and $I'$
must have the same ending point $t'$ by construction. Since $s \leq  s'_1$, it is enough to show that $\frac{t -
s'_1+1}{t - s_1 +1} \geq \lambda$ and it follows from the given condition that $|[s'_1, t_1]| \geq \lambda |[s_1,
t_1]|$, (i.e. $t_1 - s'_1 + 1 \geq \lambda(t_1 - s_1 +1)$) and $t \geq t_1$.
\end{proof}

\end{document}